\documentclass[a4paper]{article}

\usepackage{a4wide,amsmath,amssymb,amsthm,hyperref,multicol,tikz-cd}

\newcommand{\cod}{\mathit{cod}}
\newcommand{\cofission}[1]{{\cofissionop({#1})}}
\newcommand{\cofissionop}{{\tilde{\fis}}}
\newcommand{\cofusion}[1]{{\cofusionop({#1})}}
\newcommand{\cofusionop}{{\tilde{\fus}}}
\newcommand{\converse}{\smallsmile}
\newcommand{\convex}[1]{{#1}{\updownarrow}}
\newcommand{\dom}{\mathit{dom}}
\newcommand{\down}[1]{{#1}{\downarrow}}
\newcommand{\dual}[1]{{#1}^\mathsf{d}}
\newcommand{\fboxalp}[1]{{[ {#1} ]_\alpha}}
\newcommand{\fboxast}[1]{{[ {#1} ]_\seq}}
\newcommand{\fboxr}[1]{{[ {#1} ]_r}}
\newcommand{\fdiaalp}[1]{{\langle {#1} \rangle_\alpha}}
\newcommand{\fdiaast}[1]{{\langle {#1} \rangle_\seq}}
\newcommand{\fdiar}[1]{{\langle {#1} \rangle_r}}
\newcommand{\FF}{\mathcal{F}}
\newcommand{\fis}{\delta_i}
\newcommand{\fus}{\delta_o}
\newcommand{\gfboxalp}[1]{{[ {#1} ]_\alpha^\GG}}
\newcommand{\gfboxast}[1]{{[ {#1} ]_\seq^\GG}}
\newcommand{\gfboxr}[1]{{[ {#1} ]_r^\GG}}
\newcommand{\gfdiaalp}[1]{{\langle {#1} \rangle_\alpha^\GG}}
\newcommand{\gfdiaast}[1]{{\langle {#1} \rangle_\seq^\GG}}
\newcommand{\gfdiar}[1]{{\langle {#1} \rangle_r^\GG}}
\newcommand{\GG}{\mathcal{G}}
\newcommand{\icpl}[1]{{\sim}{#1}}
\newcommand{\Id}{\mathit{Id}}
\newcommand{\ii}{\Cap}
\newcommand{\iiatoms}{{\mathsf{A}_\ii}}
\newcommand{\iione}{{1_\ii}}
\newcommand{\iu}{\Cup}
\newcommand{\iU}{\raisebox{-0.5ex}{\Large$\iu$}}
\newcommand{\iuatoms}{{\mathsf{A}_\iu}}
\newcommand{\iuone}{{1_\iu}}
\newcommand{\kleisli}{\Pow}
\newcommand{\Mult}{M}
\newcommand{\Pow}{\mathcal{P}}
\newcommand{\Rel}{\mathbf{Rel}}
\newcommand{\RR}{\mathcal{R}}
\newcommand{\rto}{\leftrightarrow}
\newcommand{\seq}{\ast}
\newcommand{\seqint}{\odot}
\newcommand{\Set}{\mathbf{Set}}
\newcommand{\subem}{\mathrel{\sqsubseteq_\updownarrow}}
\newcommand{\subh}{\mathrel{\sqsubseteq_\downarrow}}
\newcommand{\subs}{\mathrel{\sqsubseteq_\uparrow}}
\newcommand{\syq}[2]{{#1} \div {#2}}
\newcommand{\up}[1]{{#1}{\uparrow}}

\newtheorem{theorem}{Theorem}[section]
\newtheorem{proposition}[theorem]{Proposition}
\newtheorem{lemma}[theorem]{Lemma}
\newtheorem{corollary}[theorem]{Corollary}

\theoremstyle{definition}
\newtheorem{example}[theorem]{Example}
\newtheorem{remark}[theorem]{Remark}

\begin{document}

\title{Modal Algebra of Multirelations}
\author{Hitoshi Furusawa, Walter Guttmann and Georg Struth}
\maketitle

\begin{abstract}
  We formalise the modal operators from the concurrent dynamic logics of Peleg, Nerode and Wijesekera in a multirelational algebraic language based on relation algebra and power allegories, using relational approximation operators on multirelations developed in a companion article.
  We relate Nerode and Wijesekera's box operator with a relational approximation operator for multirelations and two related operators that approximate multirelations by different kinds of deterministic multirelations.
  We provide an algebraic soundness proof of Goldblatt's axioms for concurrent dynamic logic as an application.
\end{abstract}

\section{Introduction}
\label{section.introduction}

This is the third article in a trilogy on the inner structure of multirelations~\cite{FurusawaGuttmannStruth2023a}, the determinisation of such relations~\cite{FurusawaGuttmannStruth2023b} and their algebras of modal operators.

Multirelations are relations of type $X \rto \Pow Y$, which model alternating nondeterminism.
Here we contribute to a line of work on modal algebras of multirelations~\cite{FurusawaStruth2015a,FurusawaStruth2016} related to Peleg's concurrent dynamic logic~\cite{Peleg1987} and algebraic languages for these~\cite{FurusawaKawaharaStruthTsumagari2017}.
These languages are extensions of relation algebra~\cite{Schmidt2011} and boolean power allegories~\cite{FreydScedrov1990} with specific operations for multirelations.

Our main motivation has been the algebraic formalisation of Nerode and Wijesekera's modal box operator~\cite{NerodeWijesekera1990} for concurrent dynamic logic.
In dynamic logics, relational modal box operators
\begin{equation*}
  \fboxr{S} Q = \{ a \in X \mid \forall b \in Y .\ S_{a,b} \Rightarrow b \in Q \},
\end{equation*}
for the relation $S : X \rto Y$ and the set $Q \subseteq Y$, typically express correctness specifications of relational programs: $\fboxr{S} Q$ determines the set of states from which every terminating execution of program $S$ must be in the set $Q$.
Yet Peleg's multirelational box operator
\begin{equation*}
  \fboxast{R} Q = \{ a \mid \forall B .\ R_{a,B} \Rightarrow B \cap Q \neq \emptyset \}
\end{equation*}
fails to capture the correctness of the multirelational program $R : X \rto \Pow Y$ in the presence of nonterminating elements of the form $(a,\emptyset)$~\cite{NerodeWijesekera1990}.
Nerode and Wijesekera therefore propose an alternative multirelational box
\begin{equation*}
  \fboxalp{R} Q = \{ a \mid \forall B .\ R_{a,B} \Rightarrow B \subseteq Q \};
\end{equation*}
Goldblatt has subsequently reduced it to a relational box, $\fboxalp{R} Q = \fboxr{\alpha(R)} Q$, approximating the multirelation $R : X \rto \Pow Y$ by the relation
\begin{equation*}
  \alpha(R) : X \rto Y, R \mapsto \{ (a,b) \mid b \in \bigcup \{ B \mid (a,B) \in R \} \},
\end{equation*}
in which the alternating structure of $R$ has been atomised.
But how can $\fboxalp{-}$ and $\alpha(-)$ be formalised in an algebraic multirelational language?
Which operations and constructs are needed for expressing it and for studying its relationship to $\fboxast{-}$?

A key observation is that Goldblatt's $\alpha : (X \rto \Pow Y) \to (X \rto Y)$ can be expressed using fundamental concepts of power allegories.
For any multirelation $R : X \to \Pow Y$,
\begin{equation*}
  \alpha(R) = R{\ni},
\end{equation*}
the relational composition of $R$ with the converse of the element-(multi)relation ${\in} : Y \rto \Pow Y$.
Likewise we can define a new De Morgan dual diamond operator $\fdiaalp{-} = \fdiar{-} \circ \alpha$.
Yet the combination of $\fboxalp{-}$ with Peleg's multirelational diamond $\fdiaast{-}$, which is used in Wijesekera and Nerode's as well as in Goldblatt's concurrent dynamic logics, and the study of the relationships between the different modal operators requires further multirelational operations, including Peleg's composition of multirelations.
Due to the complexity of their interactions we consider them in concrete extensions and enrichments of the category $\Rel$ in this work, but with a view towards future axiomatic approaches, and therefore by and large through algebraic proofs.

Modal operators on relations or multirelations usually map relations $X \rto Y$ or multirelations $X \rto \Pow Y$ to functions or ``predicate transformers'' $\Pow Y \to \Pow X$.
Modal diamonds then arise as relational or multirelational preimage operations and modal boxes as their De Morgan duals.
Transformers $\Pow X \to \Pow Y$ are obtained by opposition.
Algebraically, these preimages can be expressed using relational and multirelational domain operations together with relational or Peleg composition.
Alternatively, via the isomorphism between relations $X \rto Y$ and functions $X \to \Pow Y$, predicate transformers $\Pow X \to \Pow Y$ can be obtained as Kleisli extensions of Kleisli arrows in the powerset monad on $\Set$.
Similarly, we can map multirelations $X \rto \Pow Y$ to ``relational predicate transformers'' $\Pow X \rto \Pow Y$ using a Kleisli lifting for multirelations introduced in~\cite{FurusawaKawaharaStruthTsumagari2017}, while relations $X \rto Y$ are sent to such transformers by the relational image functor of power allegories.
Predicate transformers $\Pow Y \rto \Pow X$ can again be obtained via relational converse.
The graph functor translates these transformers from $\Set$ to $\Rel$.

While these constructions of modalities depend only on concepts of power allegories and on Peleg composition, the relationships between the different modal operators require closure and duality properties of the inner structure of multirelations and notions of inner determinism and inner functionality (or univalence), which have been studied in the first two parts of this trilogy~\cite{FurusawaGuttmannStruth2023a,FurusawaGuttmannStruth2023b}.
Here we harvest the concepts and results sown in the previous parts to obtain the results outlined above.
As an application, we prove soundness of a variant of Goldblatt's axiomatisation of concurrent dynamic logic.

As in~\cite{FurusawaGuttmannStruth2023a,FurusawaGuttmannStruth2023b}, we have used the Isabelle/HOL proof assistant to formalise and check many results in this article, see~\cite{GuttmannStruth2023}, without aiming at a complete formalisation.
Our article is therefore self-contained without the Isabelle libraries.


\section{Relations and Multirelations}
\label{section.relation-multirelation}

We start with recalling the basics of binary relations and multirelations, following~\cite{FurusawaGuttmannStruth2023a,FurusawaGuttmannStruth2023b}; see also the references therein.
Our multirelational language is closely related to allegorical and relation-algebraic approaches~\cite{SchmidtStroehlein1989,Schmidt2011,FreydScedrov1990,BirdMoor1997,FurusawaKawaharaStruthTsumagari2017}.
We also use many properties from~\cite{FurusawaGuttmannStruth2023a,FurusawaGuttmannStruth2023b}.
We list all relational and multirelational concepts with respect to a small basis in Appendix~\ref{section.basis}, which extends similar lists in the two predecessor articles.

\subsection{Binary relations}
\label{subsection.binary-relations}

Following~\cite{FurusawaGuttmannStruth2023a,FurusawaGuttmannStruth2023b}, we work in enrichments of the category $\Rel$.
We write $X \rto Y$ for the homset $\Rel(X,Y)$, $\Id_X$ for the identity relation on $X$, $\emptyset_{X,Y}$ for the least and $U_{X,Y}$ for the greatest element in $X \rto Y$, $-R$ for the complement of $R$ and $S-R$ for the relative complement $S \cap -R$, $R S$ for the relational composition of relations $R$, $S$ of suitable type, $R / S$ and $R \backslash S$ for the left and right residuals of $R$ and $S$, and $R^\converse$ for the converse of $R$.
We frequently need the modular law $R S \cap T \subseteq (R \cap T S^\converse) S$ and the properties $T \backslash S = (S^\converse / T^\converse)^\converse$, $T / S = -(-T S^\converse)$ and $T \backslash S = -(T^\converse (-S))$ of residuals.

We also need the following concepts: the \emph{symmetric quotient} $\syq{T}{S} : X \rto Y$, defined as $\syq{T}{S} = (T \backslash S) \cap (T^\converse / S^\converse)$, \emph{tests}, which are relations $R \subseteq \Id$, and whose relational composition is intersection, and the \emph{domain} map
\begin{equation*}
  R : X \rto Y, \, R \mapsto \Id_X \cap R R^\converse = \Id_X \cap R U_{Y,X} = \{ (a,a) \mid \exists b .\ R_{a,b} \}.
\end{equation*}
Domain elements and tests form the same full subalgebra of $\Rel(X,X)$ for any $X$, a complete atomic boolean algebra.
The boolean complement of a test $P$ is $\neg P = \Id - P$.

Deterministic relations play an important role in our work.
The relation $R : X \rto Y$ is \emph{total} if $\dom(R) = \Id_X$, or equivalently $\Id_X \subseteq R R^\converse$, \emph{univalent}, or a \emph{partial function}, if $R^\converse R \subseteq \Id_Y$, and \emph{deterministic}, or a \emph{function}, if it is total and univalent.
Functions as deterministic relations in $\Rel$ are of course graphs of functions in $\Set$.
We need the relational law $P Q \cap S = (P \cap S Q^\converse) Q$ for univalent $Q$~\cite{SchmidtStroehlein1989} in calculations.

Our proofs about modal operators are strongly based on concepts from power allegories, and in particular monadic concepts in relational form~\cite{FreydScedrov1990,BirdMoor1997}.
We summarise these concepts in the following; see~\cite{FurusawaGuttmannStruth2023b} for details.
The standard isomorphism between relations in $X \rto Y$ and nondeterministic functions in $X \to \Pow Y$ in $\Set$ can be expressed in $\Rel$ by taking graphs.
The \emph{power transpose} $\Lambda(R) = \syq{R^\converse}{\in_Y} = \{ (a,R(a)) \mid a \in X \}$ maps relations $X \rto Y$ to functions in $X \rto \Pow Y$, where $\in_Y : Y \rto \Pow Y$ is the membership relation on $Y$.
Conversely, relational postcomposition with its converse, the \emph{has-element relation} $\ni_Y : \Pow Y \rto Y$, maps relations in $X \rto \Pow Y$ to relations in $X \rto Y$.
We henceforth write $\alpha = (-){\ni}$.
This function satisfies $\alpha(R) = \{ (a,b) \mid b \in \bigcup R(a) \}$.

The \emph{relational image functor} $\Pow : (X \rto Y) \to (\Pow X \rto \Pow Y), \, R \mapsto \Lambda({\ni_X} R)$, which satisfies $\Pow(R) = \{ (A,R(A)) \mid A \subseteq X \}$, codes the relational image, given by the covariant powerset functor in $\Set$, again as a graph.
It is deterministic by definition.
The unit and multiplication of the monad of the powerset functor in $\Set$ are recovered in $\Rel$ as $\eta_X : X \rto \Pow X$ and $\mu_X : \Pow^2 X \rto \Pow X$ such that $\eta_X = \Lambda (\Id_X) = \{ (a,\{a\}) \mid a \in X \}$ and $\mu_X = \Pow({\ni_X})$.

Apart from tests we need the \emph{power test} $P_\seq : \Pow X \rto \Pow X$ of any test $P \subseteq \Id_X$, which is defined as $P_\seq = ({\in_X} \backslash P {\in_X}) \cap \Id_{\Pow X}$~\cite{FurusawaKawaharaStruthTsumagari2017}.
Equivalently, it can be expressed as $P_\seq = ({\in_X} \backslash P U_{X,\Pow X}) \cap \Id_{\Pow X} = \{ (A,A) \mid \forall a \in A .\ (a,a) \in P \}$, and it holds that $P_\seq \subseteq \Id_{\Pow X} = (\Id_X)_\seq$.

Finally, we need the \emph{subset relation} $\Omega_Y = {\in_Y} \backslash {\in_Y} = \{ (A,B) \mid A \subseteq B \subseteq Y \}$ and the \emph{complementation relation} $C = \syq{\in_Y}{-{\in_Y}} = \{ (A,-A) \mid A \subseteq Y \}$ to manipulate multirelations.

\subsection{Multirelations}
\label{SS:multirelations}

A \emph{multirelation} is an arrow $X \rto \Pow Y$ in $\Rel$ and therefore a doubly-nondeterministic function $X \to \Pow^2 Y$ in $\Set$.
We write $\Mult(X,Y)$ for the homset $X \rto \Pow Y$.
Multirelations allow two levels of nondeterminism and two level of choices: an outer or angelic level given two elements $(a,B)$ and $(a,C)$ of a multirelation, and an inner or demonic level given by the elements $b \in B$ for any $(a,B)$.

The \emph{Peleg composition}~\cite{Peleg1987} $\seq : (X \rto \Pow Y) \times (Y \rto \Pow Z) \to (X \rto \Pow Z)$ can be defined in two steps from the \emph{Kleisli lifting} $(-)_\kleisli : (X \rto \Pow Y) \to (\Pow X \rto \Pow Y)$ and the \emph{Peleg lifting} $(-)_\seq : (X \rto \Pow Y) \to (\Pow X \rto \Pow Y)$ of multirelations~\cite{FurusawaKawaharaStruthTsumagari2017}:
\begin{equation*}
  R_\kleisli = \Pow(\alpha(R)), \qquad
  R_\seq = \dom(R)_\seq \bigcup_{S \subseteq_d R} S_\kleisli, \qquad
  R \seq S = R S_\seq,
\end{equation*}
where, for $R, S : X \rto Y$, $S \subseteq_d R$ if $S$ is univalent, $\dom(S) = \dom(R)$ and $S \subseteq R$.
Expanding definitions,
\begin{align*}
  R_\kleisli & = \left\{ (A,B) \mid B = \bigcup R(A) \right\}, \\
  R_\seq & = \left\{ (A,B) \mid \exists f : X \to \Pow Y .\ f|_A \subseteq R \wedge B = \bigcup f(A) \right\}, \\
  R \seq S & = \left\{ (a,C) \mid \exists B .\ R_{a,B} \wedge \exists f : Y \to \Pow Z .\ f|_B \subseteq S \wedge C = \bigcup f(B) \right\}.
\end{align*}

The Kleisli lifting is the multirelational analogue of the Kleisli lifting or Kleisli extension in the Kleisli category of the powerset monad; see~\cite{FurusawaGuttmannStruth2023b} for more details.
It can also be seen as the relational image of the relational approximation of a given multirelation using the map $\alpha$.
By definition, Kleisli liftings of multirelations are functions in $\Rel$.
Peleg and Kleisli liftings coincide on deterministic multirelations: $R_\seq = R_\kleisli$ if $R$ is deterministic~\cite{FurusawaKawaharaStruthTsumagari2017}.

The units of Peleg composition are given by the multirelations $\eta_X$; because of this, we henceforth also write $1_X$ for them.
It is important to note that Peleg composition is not associative -- only $(R \seq S) \seq T \subseteq R \seq (S \seq T)$ holds -- so that multirelations do not form a category under Peleg composition.
Yet it becomes associative if the third factor is univalent~\cite{FurusawaKawaharaStruthTsumagari2017}.

Peleg composition also preserves arbitrary unions in its first argument.
It therefore has a right adjoint~\cite{FurusawaGuttmannStruth2023a}
\begin{equation*}
  R \seq S = R S_\seq \subseteq T \Leftrightarrow R \subseteq T / S_\seq = T / S.
\end{equation*}
where residual notation has been overloaded.

Multirelational modal operators require multirelational tests.
These are subsets of $1_X$ for any $X$, which are related to relational tests by the isomorphism $(-) 1_X$ from relational tests to multirelational ones and its inverse $(-) 1_X^\converse$, which specialise $\Lambda$ and $\alpha$.
The boolean complement of a multirelational test $P$ is $\neg P = 1 - P$, overloading notation for the complement of relational tests.

The Peleg lifting of a multirelational test $P \subseteq 1_X$ is $P_\seq = \{ (A,A) \mid \forall a \in A .\ (a,\{a\}) \in P \}$, the power test of the isomorphic test below $\Id_X$.
This justifies overloading the power test and Peleg lifting operation.

\begin{lemma}[\cite{FurusawaKawaharaStruthTsumagari2017}]
  \label{lemma.seq-props}
  Let $R : X \rto \Pow Y$ and $P \subseteq \Id_X$.
  Then
  \begin{equation*}
    \dom(R_\seq) = \dom(R)_\seq, \qquad
    (P R)_\seq = P_\seq R_\seq, \qquad
    1_X P_\seq = P 1_X.
  \end{equation*}
\end{lemma}

\begin{remark}
  To simplify notation, we often identify relational and multirelational tests with sets.
  A power test for $P$ is then just the powerset of $P$: $P_\seq = \Pow P$.
  In particular we may assume that $\dom(R) = \{ a \mid \exists b .\ R_{a,b} \}$ for all $R : X \rto \Pow Y$.
  We write, for instance, $A \subseteq P$ instead of $\Id_A \subseteq P$ and $A \cap P = \emptyset$ instead of $\Id_A \cap P = \emptyset$.
  We write $\neg$ for set complement.
\end{remark}

We finish this section with a brief discussion of concepts related to the inner structure of multirelations, which have been studied in detail in~\cite{FurusawaGuttmannStruth2023a}, based on previous work in~\cite{Rewitzky2003,FurusawaStruth2015a,FurusawaStruth2016}.
We will henceforth speak of inner and outer concepts for a clear distinction.

For $R, S : X \rto \Pow Y$, we define the \emph{inner union} $R \iu S = \{ (a,A \cup B) \mid R_{a,A} \wedge S_{a,B} \}$ with unit $\iuone = \{ (a,\emptyset) \mid a \in X \}$, the \emph{inner complementation} $\icpl{R} = R C = \{ (a,-A) \mid R_{a,A} \}$ and the set $\iuatoms = \{ (a,\{b\}) \mid a \in X \wedge b \in Y \}$ in $\Mult(X,Y)$ of \emph{atoms}.
The inner union is associative and commutative, but need not be idempotent.

Inner deterministic multirelations are once again important.
The multirelation $R : X \rto Y$ is \emph{inner total} if $R \subseteq -\iuone$, that is, $B$ is non-empty for each $(a,B) \in R$, \emph{inner univalent} if $R \subseteq \iuatoms \cup \iuone$, that is, $B$ is either a singleton or empty for each $(a,B) \in R$, and \emph{inner deterministic} if it is inner total and inner univalent, in which case $B \subseteq Y$ is a singleton set whenever $R_{a,B}$ for some $a \in X$.
We write $\nu(R)$ for the inner total part of $R$: those pairs in $R$ whose second component is not $\emptyset$, that is, $\nu(R) = R - \iuone$.

Finally, we need the following closures and preorders on the inner structure.
For $R : X \rto \Pow Y$, the \emph{up-closure} $\up{R} = R \Omega = \{ (a,A) \mid \exists (a,B) \in R .\ B \subseteq A \}$ with preorder $R \subs S \Leftrightarrow S \subseteq \up{R}$ and the \emph{down-closure} $\down{R} = R \Omega^\converse = \{ (a,A) \mid \exists (a,B) \in R .\ A \subseteq B \}$ with preorder $R \subh S \Leftrightarrow R \subseteq \down{S}$.
The up-closure and down-closure are related by inner duality.


\section{Deterministic Multirelations}
\label{section.determinism}

In this section, we discuss more advanced properties of deterministic multirelations that are important for reasoning with modal operators.

Recall that a \emph{quantaloid} is a category in which every homset forms a complete lattice and where arrow composition preserves arbitrary sups in both arguments.

The main result in~\cite{FurusawaGuttmannStruth2023b} on inner and outer deterministic multirelations is as follows.

\begin{proposition}[{\cite[Proposition 3.8]{FurusawaGuttmannStruth2023b}}]
  \label{proposition.peleg-rel}
  The inner deterministic multirelations with $\seq$, $1$, $\bigcup$ and the outer deterministic multirelations with $\seq$, $1$ and $\iU$ form quantaloids isomorphic to the quantaloid of binary relations.
\end{proposition}

For us, it is particularly important to note that $\Lambda$ is the isomorphism from the quantaloid $\Rel$ to the quantaloid of outer deterministic multirelations that sends relational composition to Peleg composition, each $\Id_X$ to $1_X$ and unions to inner unions, which are idempotent on deterministic multirelations~\cite[Lemma 3.6]{FurusawaGuttmannStruth2023a}.
Its inverse is $\alpha$.
Similarly, $\eta = (-) 1$ is the isomorphism from $\Rel$ to the quantaloid of inner deterministic multirelations and $\alpha$ is its inverse.
The following consequence of this fact is important for concurrent dynamic logic.

\begin{lemma}[{\cite{FurusawaGuttmannStruth2023b}}]
  \label{lemma.alpha-props}
  Let $R : X \rto \Pow Y$ and $S : Y \rto \Pow Z$.
  Then $\alpha(R \seq S) \subseteq \alpha(R) \alpha(S)$, and equality holds if $R$ and $S$ are inner or outer deterministic.
\end{lemma}

\begin{example}[{\cite{FurusawaGuttmannStruth2023b}}]
  \label{example.alpha-counter}
  For $R = \{ (a,\{a,b\}) \}$,
  \begin{equation*}
    \alpha(R \seq R) = \alpha(\emptyset) = \emptyset \subset \{ (a,a), (a,b) \} = \alpha(R) = \alpha(R) \alpha(R).
  \end{equation*}
\end{example}

Note also that the preorders $\subh$ and $\subs$ are partial orders on inner and outer deterministic multirelations.
They coincide on outer deterministic multirelations, and $\subh$ and $\sqsupseteq_\uparrow$ coincide with $\subseteq$ on inner deterministic multirelations~\cite[Proposition 5.8]{FurusawaGuttmannStruth2023a}.

In~\cite{FurusawaGuttmannStruth2023b}, we have also introduced determinisation maps for multirelations.
Let $R : X \rto \Pow Y$ be a multirelation.
The \emph{outer determinisation} operation $\fus = \Lambda \circ \alpha$ sends $R$ to the outer deterministic multirelation isomorphic to relation $\alpha(R)$.
The \emph{inner determinisation} operation $\fis = \eta \circ \alpha$ sends $R$ to the inner deterministic multirelation isomorphic to $\alpha(R)$.
They satisfy
\begin{equation*}
  \fus(R) = \{ (a,B) \mid B = \bigcup R(a) \} \qquad
  \text{ and } \qquad
  \fis(R) = \{ (a,\{b\}) \mid b \in \bigcup R(a) \}.
\end{equation*}

\begin{proposition}[{\cite[Corollary 3.10]{FurusawaGuttmannStruth2023b}}]
  \label{cor.det-cat2}
  The maps $\fis$ and $\fus$ are isomorphisms between the quantaloids of inner deterministic and outer deterministic multirelations.
\end{proposition}

By functoriality, $\fis(R \seq S) = \fis(R) \seq \fis(S)$ if $R$, $S$ are outer deterministic and, similarly, $\fus(R \seq S) = \fus(R) \seq \fus(S)$ if $R$, $S$ are inner deterministic.
Moreover, the inner and outer deterministic multirelations are precisely the fixpoints of $\fis$ and $\fus$, respectively~\cite[Corollary 3.11]{FurusawaGuttmannStruth2023b}.

We end this section with a collection of properties that are helpful in proofs in the remaining sections.

\begin{lemma}
  \label{lemma.lambda-alpha-props}
  Let $R : X \to Y$, $S : Y \rto Z$ and let $f : X \to \Pow Y$ be outer deterministic.
  Then
  \begin{enumerate}
  \item $\alpha(\Lambda(R)) = R$ and $\Lambda(\alpha(f)) = f$,
  \item $f \Lambda(R) = \Lambda(f R)$ and $\Lambda(\ni_X) = \Id_{\Pow X}$.
  \item $\Lambda(R S) = \Lambda(R) \Pow(S)$,
  \item $\Lambda(f) = f \eta$,
  \item $\alpha(\eta) = \Id$.
  \end{enumerate}
\end{lemma}


\section{Properties of Tests and Domain}
\label{section.tests-domain}

Before turning to modal operators, we collect properties of tests and domain.
Recall from Section~\ref{section.relation-multirelation} that we usually identify tests with sets and do not consider multirelational tests as subsets of $1$ explicitly.
More specifically, Peleg compositions $P \seq R$ and $R \seq P$ of multirelational tests $P$ and multirelations $R$ can always be replaced by relational compositions $\alpha(P) R$ and $R P_\seq$, respectively, where $\alpha(P) = P 1^\converse$ is the isomorphic relational test and $P_\seq$ the isomorphic power test of $P$.

\begin{lemma}
  \label{lemma.test-u}
  Let $R : X \rto \Pow Y$ and $P \subseteq Y$.
  Then
  \begin{enumerate}
  \item $R P_\seq = R \cap U P_\seq$ and $R \neg (P_\seq) = R - U P_\seq$,
  \item $(\down{1})_\seq P_\seq = \{ (A,B) \mid B \subseteq A \cap P \}$,
  \item $\down{(U P_\seq)} = U P_\seq$.
  \item $\nu(R P_\seq) = \nu(R) P_\seq$,
  \item $\alpha(R) \eta(P) = \fis(R) P_\seq = \fis(\down{R} P_\seq) = \fis(\nu(\down{R}) P_\seq)$.
  \end{enumerate}
\end{lemma}

\begin{proof}
  Item (1) is standard in relation algebra, (2) follows from expanding definitions.

  For (3), it suffices to show that $P_\seq \Omega^\converse \subseteq U P_\seq$.
  Since
  \begin{equation*}
    U P_\seq = U ({\in} \backslash P U \cap \Id) = ({\in} \backslash P U)^\converse = U P / {\ni},
  \end{equation*}
  this follows from $P_\seq \Omega^\converse {\ni} \subseteq U P$.
  Since $\Omega^\converse {\ni} = {\ni}$, it remains to show $P_\seq {\ni} \subseteq U P$.
  The latter follows from ${\in} P_\seq \subseteq {\in} ({\in} \backslash P U) = P U$.

  For (4), $\nu(R P_\seq) = R \cap U P_\seq \cap - \iuone = (R - \iuone) \cap U P_\seq = \nu(R) P_\seq$.

  For (5), $\alpha(R) P 1_Y = \alpha(R) 1_X P_\seq = \fis(R) P_\seq$ yields the first step, expanding definitions.
  For the second one, $\fis(R) P_\seq = \down{R} \cap U P_\seq \cap \iuatoms = \down{(\down{R} \cap U P_\seq)} \cap \iuatoms = \fis(\down{R} P_\seq)$ because $U P_\seq$ is down-closed by (3) and intersections of down-closed sets are down-closed, by general properties of closure operators.
  For the last one, $\fis(\down{R} P_\seq) = \fis(\nu(\down{R} P_\seq)) = \fis(\nu(\down{R}) P_\seq) $, using $\fis \circ \nu = \fis$~\cite[Lemma 4.1(2)]{FurusawaGuttmannStruth2023b} in the first and (4) in the last step.
\end{proof}

Tests $P \subseteq \Id_X$ obviously form a boolean algebra with complementation $\neg P = \Id_X - P$.
Tests $P_\seq \subseteq \Id_{\Pow X}$ thus form a boolean algebra with $\neg (P_\seq) = \Id_{\Pow X} - P_\seq$.

Domain operations on algebras of multirelations have been studied in~\cite{FurusawaStruth2015a,FurusawaStruth2016}.
The standard definition for relations in $X \rto Y$ in Section~\ref{section.relation-multirelation} covers multirelations.

Recall that $\dom(R_\seq) = \dom(R)_\seq$ by Lemma~\ref{lemma.seq-props}; also note that $\dom(P) = P = \dom(\eta(P))$ for tests $P$.
Another important property of domain is \emph{domain locality}~\cite{FurusawaStruth2015a},
\begin{equation*}
  \dom(R S_\seq) = \dom(R \, \dom(S)_\seq),
\end{equation*}
in particular for the definition of modal operators in Section~\ref{section.modal}.
Further, $\dom(P R_\seq) = P \cap \dom(R)$, and, set-theoretically,
\begin{equation*}
  \dom(R P_\seq) = \{ a \mid \exists B .\ R_{a,B} \wedge B \subseteq P \} \qquad
  \text{and} \qquad
  \dom(R_\seq) = \{ A \mid \forall a \in A .\ \exists B .\ R_{a,B} \}.
\end{equation*}

\begin{lemma}
  \label{lemma.dom-det}
  Let $R : X \rto \Pow Y$ and $P \subseteq Y$.
  Then
  \begin{enumerate}
  \item $\dom(\fis(R)) = \dom(\alpha(R)) = \dom(\alpha(\nu(R))) = \dom(\nu(R)) = \dom(\nu(\fus(R)))$,
  \item $\dom(\alpha(R) P) = \dom(\nu(\down{R}) P_\seq)$ and $\dom(\alpha(R) \neg P) = \dom(R \neg (P_\seq))$.
  \end{enumerate}
\end{lemma}

\begin{proof}
  For (1), first $\dom(\fis(R)) = \dom(\alpha(R) 1) = \dom(\alpha(R) \dom(1)) = \dom(\alpha(R))$.
  Second, we obtain $\dom(\alpha(R)) = \dom(\alpha(\nu(R)))$ because $\alpha = \alpha \circ \nu$~\cite[Lemma 4.1(1)]{FurusawaGuttmannStruth2023b}.
  Third,
  \begin{equation*}
    \dom({\ni}) = \dom({\ni} \dom(U)) = \dom({\ni} U) = \dom(-\iuone^\converse)
  \end{equation*}
  since $U {\in} = -\iuone$.
  Hence
  \begin{align*}
        \dom(\alpha(\nu(R)))
    & = \dom(\nu(R) \dom({\ni})) \\
    & = \dom(\nu(R) \dom(-\iuone^\converse)) \\
    & = \dom(\nu(R) (-\iuone^\converse)) \\
    & = \dom((\nu(R) - \iuone) U) \\
    & = \dom(\nu(R)).
  \end{align*}
  Fourth, $\dom(\nu(\fus(R))) = \dom(\fis(\fus(R))) = \dom(\fis(R))$ using the previous identities and the fact that $\fis \circ \fus = \fis$~\cite[Lemma 3.17]{FurusawaGuttmannStruth2023b}.

  For the first equality in (2), we have
  \begin{equation*}
      \dom(\alpha(R) P)
    = \dom(\alpha(R) P 1)
    = \dom(\fis(\down{R} P_\seq))
    = \dom(\nu(\down{R} P_\seq))
    = \dom(\nu(\down{R}) P_\seq)
  \end{equation*}
  using Lemma~\ref{lemma.test-u} and part (1) of this lemma.
  For the second,
  \begin{equation*}
    \neg \dom({\ni} \neg P) = \Id - ({\ni} \neg P U) = \Id - ({\ni} - P U) = \Id \cap ({\in} \backslash P U) = P_\seq,
  \end{equation*}
  hence $\dom(\alpha(R) \neg P) = \dom(R \, \dom({\ni} \neg P)) = \dom(R \neg (P_\seq))$.
\end{proof}


\section{Modal Operators}
\label{section.modal}

We now turn to modal box and diamond operators for multirelations, but first recall the standard relational modalities.

\subsection{Relational modal operators}

The backward relational diamond operator $\langle - | : (X \rto Y) \to (\Pow X \to \Pow Y)$ can be obtained directly from the relational image operation: $\langle R | P = \cod(P R)$ for all $R : X \rto Y$ and $P \subseteq X$.
Its opposite forward relational diamond operator $| - \rangle : (X \rto Y) \to (\Pow Y \to \Pow X)$ is given as $| R \rangle = \langle R^\converse |$ so that $| R \rangle Q = \dom(R Q)$ for all $R : X \to Y$ and $Q \subseteq Y$.
Forward and backward box operators $| - ] : (X \rto Y) \to (\Pow Y \to \Pow X)$ and $[ - | : (X \rto Y) \to (\Pow Y \to \Pow X)$ can be obtained from $| - \rangle$ and $\langle - |$ by De Morgan duality $| R ] = \neg \circ | R \rangle \circ \neg$ and $[ R | = \neg \circ \langle R | \circ \neg$.
They are again related by opposition $| R ] = [ R^\converse |$.

Using nondeterministic functions -- arrows of the Kleisli category $\Set_\Pow$ of the powerset monad in $\Set$ -- instead, $\langle - | = (-)_K$, the Kleisli extension of such arrows.
The left triangle in the diagram below thus commutes, where $\FF : \Rel \to \Set_\Pow$ and $\RR : \Set_\Pow \to \Rel$ indicate the isomorphism between the two categories.
\begin{equation*}
  \begin{tikzcd}
    \Pow X \to \Pow Y \arrow[rr, "\GG"] && \Pow X \rto \Pow Y \\
    & X \rto Y \arrow[ul, swap, "\langle - |"] \arrow[ur, "\Pow"] \arrow[dl, bend left=10, "\FF"] \arrow[dr, bend left=10, "\Lambda"] & \\
    X \to \Pow Y \arrow[rr, swap, "\GG"] \arrow[uu, "(-)_K"] \arrow[ur, bend left=10, "\RR"] & & X \rto \Pow Y \arrow[uu, swap, "(-)_\Pow"] \arrow[ul, bend left=10, "\alpha"]
  \end{tikzcd}
\end{equation*}

The graph functor $\GG : \Set \to \Rel$ translates this diagram to $\Rel$.
The relational image operator $\Pow$ and the Kleisli lifting $(-)_\Pow$ in $\Rel$ now play the role of $\langle - |$ and $(-)_K$ in $\Set$.
The Kleisli extension $(-)_K$ is actually an isomorphism between $\Set_\Pow$ and the category of backward diamond operators (with inverse given by postcomposition with the unit of the powerset monad).
It preserves the entire quantaloid structure in $\Set_\Pow$ that comes from $\Rel$.
The standard definition of the Kleisli lifting translates to multirelations: $R_\kleisli = \Pow(R) \mu$~\cite{FurusawaGuttmannStruth2023b}.
Therefore, and of course owing to the composition of isomorphisms along the other outer faces of our diagram, we obtain a corresponding isomorphism between categories of outer deterministic multirelations and their Kleisli liftings on the right of this diagram.
All arrows in the diagram are therefore isomorphisms between categories.

In the following, we focus on the forward operators $| - \rangle$ and $| - ]$ and use the more conventional notation $\fdiar{-}$ and $\fboxr{-}$.
Hence, for $R : X \rto Y$ and $P \subseteq Y$,
\begin{equation*}
  \fdiar{R} Q = \dom(R Q) = \{ a \mid \exists b \in Q .\ R_{a,b} \} \qquad
  \text{ and } \qquad
  \fboxr{R} Q = \neg \dom(R \neg Q) = \{ a \mid \forall b \in Q .\ R_{a,b} \}.
\end{equation*}

For backward box and diamond operators we simply write $\fdiar{R^\converse}$ and $\fboxr{R^\converse}$.
It follows that $\fdiar{R^\converse} P = \{ b \mid \exists a \in P .\ R_{a,b} \}$ and $\fboxr{R^\converse} P = \{ b \mid \forall a \in P .\ R_{a,b} \}$.
This describes once again the left of the above diagram.

On the right of the diagram above it is routine to check that the graphs of $\fdiar{R} Q$ and $\fboxr{R} P$ are given by
\begin{align*}
  \Pow(R^\converse) = \{ (Q,P) \mid P = \fdiar{R} Q \} \qquad
  \text{ and } \qquad
  \Lambda({\ni} / R) = \{ (Q,P) \mid P = \fboxr{R} Q \}.
\end{align*}
We henceforth write $\gfdiar{R} = \Pow(R^\converse)$ and $\gfboxr{R} = \Lambda({\ni} / R)$, where the superscript ${}^\GG$ indicates that these operators are graphs of $\fdiar{R}$ and $\fboxr{R}$.
The fact that $\gfboxr{R}$ is an allegorical box operator in $\Rel$ is already known~\cite{BirdMoor1997}.

\subsection{Multirelational modal operators}
\label{subsection.new-diamond}

Two extensions of relational modalities to multirelational ones have been proposed.
The first replaces relational composition by Peleg composition in relational image operations and their De Morgan duals.
It has been advocated by Peleg in concurrent dynamic logic~\cite{Peleg1987} and studied further in~\cite{FurusawaStruth2015a,FurusawaStruth2016}.
The second has been proposed implicitly by Nerode and Wijesekera~\cite{NerodeWijesekera1990} and studied further in~\cite{Goldblatt1992}.
In the language of power allegories, it uses $\alpha$ to approximate multirelations by relations in relational boxes and diamonds.
In the remainder of this section we formalise and relate these approaches.
We also study the second approach in $\Rel$.

First we recall the multirelational modal operators $\fdiaast{-}, \fboxast{-} : (X \rto \Pow Y) \to (\Pow Y \to \Pow X)$ used in concurrent dynamic logic.
These are defined, for all $R : X \rto \Pow Y$ and $P \subseteq Y$, as
\begin{equation*}
  \fdiaast{R} P = \dom(R \seq P) \qquad
  \text{ and } \qquad
  \fboxast{R} P = \neg \dom(R \seq \neg P).
\end{equation*}
As already mentioned, these are essentially relational modalities with relational composition replaced by Peleg composition.
Indeed, $\fdiaast{R} P = \fdiar{R} P_\seq$ and $\fboxast{R} P = \fboxr{R} \neg ((\neg P)_\seq)$.
Further, $\fdiaast{R} P = \neg \fboxast{R} \neg P$ and $\fboxast{R} P = \neg \fdiaast{R} \neg P$, and, unfolding definitions,
\begin{equation*}
  \fdiaast{R} P = \{ a \mid \exists B .\ R_{a,B} \wedge B \subseteq P \} \qquad
  \text{and} \qquad
  \fboxast{R} P = \{ a \mid \forall B .\ R_{a,B} \Rightarrow B \cap P \neq \emptyset \}.
\end{equation*}
Following Wijesekera and Nerode, we have already pointed out in the introduction that the semantics of $\fboxast{-}$ is thus quite different from that of its relational counterpart, and it fails to capture standard program correctness specifications that are usually associated with modal boxes.
Nerode and Wijesekera have therefore argued for an alternative multirelational box, to which we add a De Morgan dual diamond operator $\fdiaalp{-} : (X \rto \Pow Y) \to (\Pow Y \to \Pow X)$.
In the language of power allegories and for $R : X \rto \Pow Y$ and $P \subseteq Y$, the two operators are
\begin{equation*}
  \fboxalp{R} = \fboxr{\alpha(R)} \qquad
  \text{ and } \qquad
  \fdiaalp{R} = \fdiar{\alpha(R)}.
\end{equation*}
Unfolding definitions,
\begin{equation*}
  \fboxalp{R} P = \{ a \mid \forall B .\ R_{a,B} \Rightarrow B \subseteq P \} \qquad
  \text{ and } \qquad
  \fdiaalp{R} P = \{ a \mid \exists B .\ R_{a,B} \wedge B \cap P \neq \emptyset \},
\end{equation*}
which confirms in particular that $\fboxalp{R}P$ is consistent with Nerode and Wijesekera's definition.

These two operations satisfy again $\fdiaalp{R} P = \fdiar{\alpha(R)} P = \neg \fboxr{\alpha(R)} \neg P = \neg \fboxalp{R} \neg P$ and likewise $\fboxalp{R} P = \neg \fdiaalp{R} \neg P$.
In addition, Lemma~\ref{lemma.dom-det} implies that $\fboxalp{R} P = \neg \fdiar{R} \neg (P_\seq) = \fboxr{R} P_\seq$ and then De Morgan duality implies that $\fdiaalp{R} P = \fdiar{R} \neg ((\neg P)_\seq)$.

\subsection{Properties of multirelational modal operators}
\label{section.box-other}

Lemma~\ref{lemma.test-u} leads directly to a purely multirelational definition of $\fboxalp{-}$.

\begin{lemma}
  \label{lemma.box-U}
  Let $R : X \rto \Pow Y$ and $P \subseteq Y$.
  Then $\fboxalp{R} P = \neg \dom(R - U P_\seq)$.
\end{lemma}

Further, $\fboxalp{-}$ and $\fdiaalp{-}$ can be defined in various ways in terms of other modalities.

\begin{lemma}
  \label{lemma.box-nu}
  Let $R : X \rto \Pow Y$ and $P \subseteq Y$.
  Then
  \begin{enumerate}
  \item $\fboxalp{R} = \fboxast{\nu(\down{R})} = \fboxast{\fis(R)} = \fboxalp{\fus(R)} = \fboxalp{\fis(R)} = \fdiaalp{\fus(R)}$,
  \item $\fdiaalp{R} = \fdiaast{\nu(\down{R})} = \fdiaast{\fis(R)} = \fdiaalp{\fus(R)} = \fdiaalp{\fis(R)} = \fboxalp{\fus(R)}$.
  \end{enumerate}
\end{lemma}

\begin{proof}
  For (1), first $\fboxalp{R} P = \fboxr{\alpha(R)} P = \neg \dom(\nu(\down{R}) (\neg P)_\seq) = \fboxast{\nu(\down{R})} P$ using Lemma~\ref{lemma.dom-det}.
  Second, $\fboxalp{R} P = \neg \dom(\alpha(R) \neg P) = \neg \dom(\fis(R)(\neg P)_\seq)= \fboxr{\fis(R)} \neg ((\neg P)_\seq) = \fboxast{\fis(R)} P$ by Lemma~\ref{lemma.test-u}.
  The proofs of the third and fourth identity are trivial.
  Finally,
  \begin{equation*}
    \fboxalp{\fus(R)} P = \fboxr{\fus(R)} P_\seq = \fdiar{\fus(R)} P_\seq = \fdiaalp{\fus(R)} P
  \end{equation*}
  because $\fus(R)$ is outer deterministic, and the coincidence of relational boxes and diamonds is standard.

  For (2), first $\fdiaalp{R} P = \neg \fboxalp{R} \neg P = \neg \fdiaast{\nu(\down{R})} \neg P = \fdiaast{\nu(\down{R})} P$ using again Lemma~\ref{lemma.dom-det}.
  The remaining proofs follow by duality from (1).
\end{proof}

Moreover $\dom(R) = (R \seq \iuone) \iu 1$ and $\fis(R) = (R \ii U) \cap \iuatoms$, for $R \ii S = \icpl(\icpl{R} \iu \icpl{S})$, which yields another definition of $\fboxalp{-}$ in terms of multirelational operations and constants.

Lemma~\ref{lemma.box-nu} implies the following fact.

\begin{corollary}
  Let $R : X \rto \Pow Y$ and $P \subseteq Y$.
  Then $\fdiaast{\nu(R)} P \subseteq \fdiaalp{R} P$ and $\fboxast{R} P \subseteq \fboxalp{\nu(R)} P$.
\end{corollary}

Its proof is immediate from the standard fact that relational diamonds preserve $\subseteq$ in both arguments, while relational boxes reverse this order in their first argument (and preserve it in their second one), and relational representations of multirelational modal operators.

Lemma~\ref{lemma.box-nu} also indicates situations where multirelational modalities coincide.
\begin{corollary}
  Let $R : X \rto \Pow Y$.
  Then
  \begin{enumerate}
  \item $\fboxalp{R} = \fboxast{R}$ and $\fdiaalp{R} = \fdiaast{R}$ if $R$ is inner deterministic,
  \item $\fboxalp{R} = \fdiaalp{R}$ and $\fboxast{R} = \fdiaast{R}$ if $R$ is outer deterministic.
  \end{enumerate}
\end{corollary}
The proofs use properties of Lemma~\ref{lemma.box-nu} together with fixpoint properties of inner and outer deterministic multirelations and properties of $\fus$ and $\fis$ from~\cite{FurusawaGuttmannStruth2023a}.

\begin{remark}
  \label{remark.box-dia-spec}
  The multirelational boxes and diamonds specialise to relational ones.
  Let $R : X \rto Y$.
  Then $\fboxalp{\eta(R)} = \fboxr{\alpha(\eta(R))} = \fboxr{R}$ and $\fdiaalp{\eta(R)} = \fdiar{\eta(\Lambda(R))} = \fdiar{R}$ is a trivial consequence of the fact that $\eta$ and $\alpha$ form a bijective pair.
  Further, $\fdiaast{\eta(R)} P = \dom(R 1 P_\seq) = \dom(R P 1) = \dom(R P) = \fdiar{R} P$ using Lemma~\ref{lemma.seq-props}.
  The identity $\fboxast{\eta(R)} = \fboxr{R}$ then follows by duality.
\end{remark}

Here is another definition of $\fboxalp{-}$.

\begin{lemma}
  \label{lemma.box-fusion}
  Let $R : X \rto \Pow Y$ and $P \subseteq Y$.
  Then $\fboxalp{R} P = \neg \dom(\fus(R) \neg(P_\seq))$.
\end{lemma}

\begin{proof}
  $\fboxalp{R} P = \fboxr{\alpha(R)} P = \fboxr{\alpha(\fus(R))} P = \fboxr{\fus(R)} P_\seq = \neg \dom(\fus(R) \neg(P_\seq))$.
\end{proof}

Our final definition of $\fboxalp{-}$ requires a technical lemma.

\begin{lemma}
  \label{lemma.box-left-residual-aux}
  Let $R : X \rto \Pow Y$, $P \subseteq X$ and $Q \subseteq Y$.
  Then $P \subseteq \fboxalp{R} Q \Leftrightarrow P R \subseteq R Q_\seq$.
\end{lemma}

\begin{proof}
  Using $P$ and $Q$ as relational or multirelational tests depending on the context,
  \begin{align*}
                     P \subseteq \fboxalp{R} Q
     \Leftrightarrow P \subseteq \neg \dom(R \neg (Q_\seq))
     \Leftrightarrow P \dom(R \neg (Q_\seq)) \subseteq \emptyset
     \Leftrightarrow P R \neg (Q_\seq) \subseteq \emptyset
     \Leftrightarrow P R \subseteq R Q_\seq.
  \end{align*}
  The penultimate step uses a standard property of $\dom$ of relations.
\end{proof}

\begin{lemma}
  \label{lemma.box-left-residual}
  Let $R : X \rto \Pow Y$ and $P \subseteq Y$.
  Then $\fboxalp{R} P = (U P_\seq) / R \cap 1 = (R P_\seq) / R \cap 1$.
\end{lemma}

\begin{proof}
  ~
  For the first identity we use Lemma~\ref{lemma.box-left-residual-aux} and the Galois connection for residuation,
  \begin{align*}
                            Q \subseteq \fboxalp{R} P
          & \Leftrightarrow Q R \subseteq R P_\seq \\
          & \Leftrightarrow Q R \subseteq R \cap U P_\seq \\
          & \Leftrightarrow Q R \subseteq U P_\seq \\
          & \Leftrightarrow Q \subseteq (U P_\seq) / R \\
          & \Leftrightarrow Q \subseteq (U P_\seq) / R \cap 1.
  \end{align*}
  Thus $\fboxalp{R} P = (U P_\seq) / R \cap 1$.
  The proof of the second identity is similar.
\end{proof}

The preceding definitions of $\fboxalp{-}$ are useful for establishing further properties.
Above we have seen that $\fboxalp{\fis(R)} = \fboxalp{\fus(R)}$.
It turns out that this is an instance of a more general result.
To this end we consider sufficient conditions for a function $f$ to satisfy $\fboxalp{R} = \fboxalp{f(R)}$ or $\fboxr{R} = \fboxalp{f(R)}$:
\begin{enumerate}
\item Note that $\fboxalp{R} = \fboxr{\alpha(R)} = \fboxr{\alpha(f(R))} = \fboxalp{f(R)}$ if $\alpha \circ f = \alpha$.
      The latter holds, for example, for $f = \down{}$ or $f = \nu$ because $\alpha(\down{R}) = \alpha(R)$~\cite[Lemma 3.15(3)]{FurusawaGuttmannStruth2023b} and $\alpha \circ \nu = \alpha$~\cite[Lemma 4.1(1)]{FurusawaGuttmannStruth2023b}.
\item If $\fboxalp{R} = g(f(R))$ for an arbitrary function $g$ describing the context, then
      \begin{equation*}
        \fboxalp{R} = g(f(R)) = g(f(f(R))) = \fboxalp{f(R)}
      \end{equation*} if $f$ is idempotent.
      This situation arises, for example, for $g = \fboxast{-}$ and $f = \fis$ or $f = \nu(\down{-})$, and for $g = \fboxast{\nu(-)}$ and $f = \down{}$ by~\cite[Lemma 3.17]{FurusawaGuttmannStruth2023b} and Lemma~\ref{lemma.box-nu}.
      Another instance has $f = \fus$ according to Lemma~\ref{lemma.box-fusion}.
\item If $\fboxalp{R} = \fboxalp{g(R)}$ has already been established (for example, by the above instances) and $g \circ f = g$, then $\fboxalp{R} = \fboxalp{g(R)} = \fboxalp{g(f(R))} = \fboxalp{f(R)}$ generalising the first pattern.
      Combinations of $\fis$, $\fus$ and $\nu$ that give rise to instances can be found in~\cite[Lemma 3.17, 4.1]{FurusawaGuttmannStruth2023b}.
\item If $\fboxalp{R} = \fboxr{g(R)}$ has already been established and $f$ is a right-inverse of $g$, then
      \begin{equation*}
        \fboxr{R} = \fboxr{g(f(R))} = \fboxalp{f(R)}.
      \end{equation*}
      For example, $\Lambda$ is a right-inverse of $\alpha$, hence $\fboxr{R} = \fboxalp{\Lambda(R)}$ for any relation $R$.
\end{enumerate}

\begin{remark}
  The above instances imply that $\fboxalp{R} = \fboxalp{\fus(R)} = \fboxalp{\fis(R)} = \fboxalp{\down{R}} = \fboxalp{\nu(R)}$ and further equalities.
  All of the arguments $R$, $\fus(R)$, $\fis(R)$, $\down{R}$ and $\nu(R)$ contain essentially the same information (as regards box) just arranged differently with different degrees of redundancy.
  The different options have different properties: for example, $\down{\fus(-)}$ and $\down{}$ are closure operators with respect to $\subseteq$; $\nu$ is an interior operator with respect to $\subseteq$; $\fus$/$\fis$ are closure/interior operators with respect to $\subh$; other options give extrema with respect to certain (pre)orders.
\end{remark}

\subsection{Multirelational modal operators as relations}

We now return to the right-hand triangle in the diagram from the beginning of this section.
First we consider graphs of the relational modal operators $\fboxr{-}$ and $\fdiar{-}$.
As already mentioned, they are given by $\gfboxr{-} = \Lambda({\ni} / (-))$ and $\gfdiar{-} = \Pow((-)^\converse)$, respectively.
More concretely, rewriting the definitions at the beginning of this section shows that, for $R : X \rto Y$, $P \subseteq X$ and $Q \subseteq Y$,
\begin{equation*}
  \gfdiar{R} = \{ (Q,P) \mid P = \fdiar{R} Q \} \qquad
  \text{ and } \qquad
  \gfboxr{R} = \{ (Q,P) \mid P = \fboxr{R} Q \}.
\end{equation*}
Backward modalities $\langle R^\converse \rangle$ and $[ R^\converse ]$ correspond in an analogous way to $\Pow(R)$ and $\Lambda({\ni} / R^\converse)$.
The extension to the multirelational modalities $\gfdiaalp{-}$ and $\gfboxalp{-}$ is then straightforward by inserting $\alpha$'s: for $R : X \rto \Pow Y$, $P \subseteq X$ and $Q \subseteq Y$, $\gfdiaalp{R} = \gfdiar{\alpha(R)}$ and $\gfboxalp{R} = \gfboxr{\alpha(R)}$, and therefore
\begin{equation*}
  \gfdiaalp{R} = \{ (Q,P) \mid P = \fdiaalp{R} Q \} \qquad
  \text{ and } \qquad
  \gfboxalp{R} = \{ (Q,P) \mid P = \fboxalp{R} Q \}.
\end{equation*}
This explains in particular the role of the Kleisli lifting $(-)_\Pow = \Pow \circ \alpha$ as a backward diamond operator $(X \rto \Pow Y) \to (\Pow X \rto \Pow Y)$ relative to its standard counterpart $(-)_K : (X \to \Pow Y) \to (\Pow X \to \Pow Y)$ in our diagram more formally.

For any relation $X \rto Y$ or multirelation $X \rto \Pow Y$ we can use these correspondences for translating properties from the maps $\Pow X \to \Pow Y$ to their graphs $\Pow X \rto \Pow Y$, at least in tabular allegories.

Alternatively we can also reason about the latter in point-free style in the extended language of power allegories used in previous sections.
These use the subset relation $\Omega = {\in} \backslash {\in} = -(-{\ni} {\in})$ and the complement relation $C = \syq{\in}{-{\in}} = \Lambda(-{\ni})$ rather strongly.
They also use properties such as $-{\in} C = {\in}$ and ${\in} C = -{\in}$, $C = C^\converse$, $C^2 = \Id$ and in particular that $R (-S) = -(R S)$ if and only if relation $R$ is outer deterministic apart from other standard properties of the relational calculus.
We also recall a lemma from~\cite{FurusawaGuttmannStruth2023b}, to which we add a third property.

\begin{lemma}
  \label{lemma.rel-mod-props}
  Let $R : X \rto Y$.
  Then
  \begin{enumerate}
  \item $\Lambda(R) C = \Lambda(-R)$,
  \item $\Lambda(R) \Omega = R^\converse \backslash {\in} = ({\ni} / R)^\converse$.
  \item ${\ni} / R = (R {\in})^{\mathsf{d}\converse}$.
  \end{enumerate}
\end{lemma}

\begin{proof}
  Items (1) and (2) are proved in~\cite[Lemma 2.3]{FurusawaGuttmannStruth2023b}.
  For (3), $(R {\in})^{\mathsf{d}\converse} = -(C (R {\in})^\converse)=-(-{\ni} R^\converse) = {\ni} / R$, using the definition of $(-)^{\mathsf{d}}$ in the first step.
\end{proof}

\begin{corollary}
  \label{corollary.box-no-residual}
  Let $R : X \rto Y$.
  Then $\gfboxr{R} = \Lambda((R {\in})^{\mathsf{d}\converse})$ and $\gfdiar{R} = \Lambda((R {\in})^\converse)$.
\end{corollary}
The second identity is trivial, but it shows the correspondence with the first one.

Next we check the standard relationships between boxes and diamonds.
In (1) below we prove the De Morgan duality between boxes and diamonds.
In (2) and (3) we translate the standard conjugations and Galois connections for relational boxes and diamonds:
\begin{equation*}
  \fdiar{R} P \subseteq \neg Q \Leftrightarrow \fdiar{R^\converse} Q \subseteq \neg P \qquad
  \text{ and } \qquad
  \fdiar{R^\converse} P \subseteq Q \Leftrightarrow P \subseteq \fboxr{R} Q.
\end{equation*}

\begin{lemma}
  \label{lemma.box-dia-dual}
  Let $R : X \rto Y$.
  Then
  \begin{enumerate}
  \item $C \gfdiar{R} C = \gfboxr{R}$ and $\gfdiar{R} = C \gfboxr{R} C$,
  \item $C (\gfdiar{R} \Omega)^\converse = \gfdiar{R^\converse} C \Omega^\converse$,
  \item $\gfboxr{R} \Omega^\converse = (\gfdiar{R^\converse} \Omega)^\converse$.
  \end{enumerate}
\end{lemma}

\begin{proof}
  For (1), $C \Pow(R^\converse) C = \Lambda(-{\ni}) \Pow(R^\converse) C = \Lambda (-{\ni} R^\converse) C = \Lambda (-(-{\ni} R^\converse)) = \Lambda({\ni} / R)$.
  The first step unfolds $C$, the second uses Lemma~\ref{lemma.lambda-alpha-props}, the third Lemma~\ref{lemma.rel-mod-props}(1) and the fourth a standard property of residuals.
  The second identity in (1) then follows from $C^2 = \Id$.

  For (2),
  \begin{align*}
      \Pow(R) C \Omega^\converse
    & = \Lambda(-({\ni} R)) (-(-{\ni} {\in})) \\
    & = -(\Lambda(-({\ni} R)) (-{\ni}) {\in}) \\
    & = -(\Lambda(-({\ni} R)) C {\ni} {\in}) \\
    & = -(\Lambda({\ni} R) {\ni} {\in}) \\
    & = -({\ni} R {\in}) \\
    & = -(C (-{\ni}) ({\ni} R^\converse)^\converse) \\
    & = C ({\ni} / ({\ni} R^\converse)) \\
    & = C (\Pow(R^\converse) \Omega)^\converse
  \end{align*}
  The first step rewrites $\Pow$ and $\Omega^\converse$, and applies Lemma~\ref{lemma.rel-mod-props}(1).
  The second uses outer determinism of $\Lambda(-({\ni} R))$.
  The third uses outer determinism of $C$ and Lemma~\ref{lemma.lambda-alpha-props}.
  The fourth uses again Lemma~\ref{lemma.rel-mod-props}(1) and the fifth again Lemma~\ref{lemma.lambda-alpha-props}.
  The sixth step uses standard properties of $C$.
  The seventh uses outer determinism of $C$ and standard properties of residuation.
  The final step uses Lemma~\ref{lemma.rel-mod-props}(2) and properties of residuation.

  Finally, (3) is immediate from (1) and (2).
\end{proof}

The multirelational box $\gfboxalp{-}$ can be expressed using the superset relation $\Omega^\converse$ and without residuation.

\begin{lemma}
  \label{lemma.predtrans-2}
  Let $R : X \rto \Pow Y$.
  Then $\gfboxalp{R} = \Lambda(\Omega^\converse / R) = \Lambda((R {\ni} {\in}) \dual{} {}^\converse)$ and $\gfdiaalp{R} = \Lambda((R {\ni} {\in})^\converse)$.
\end{lemma}

\begin{proof}
  For the first identity, $\Omega^\converse / R = ({\ni} / {\ni}) / R = {\ni} / R {\ni} = {\ni} / \alpha(R)$, where the second step uses a standard ``currying'' property of residuals.
  The second one is immediate from Lemma~\ref{lemma.rel-mod-props}(3).
\end{proof}

The results in Lemma~\ref{lemma.box-dia-dual} translate to $\gfdiaalp{-}$ and $\gfboxalp{-}$ by instantiation with $\alpha(R)$.

Finally, we consider how the modal operators $\fdiaast{-}$ and $\fboxast{-}$ can be expressed as graphs.
For $R : X \rto \Pow Y$ it can be checked that
\begin{equation*}
  \Lambda(\up{R} {}^\converse) = \{ (Q,P) \mid P = \fdiaast{R} Q \} \qquad
  \text{ and } \qquad
  \Lambda(\up{R} \dual{} {}^\converse) = \{ (Q,P) \mid P = \fboxast{R} Q \}.
\end{equation*}
Hence $\Lambda(\up{R} {}^\converse)$ is the analog of $\fdiaast{R}$ induced by the graph functor, and $\Lambda(\up{R} \dual{} {}^\converse)$ is the graph analog of $\fboxast{R}$.
We therefore write $\gfdiaast{R} = \Lambda(\up{R} {}^\converse)$ and $\gfboxast{R} = \Lambda(\up{R} \dual{} {}^\converse)$.

First we obtain the De Morgan dualities expected.
\begin{lemma}
  \label{lemma.predtrans-3-demorgan}
  Let $R : X \rto \Pow Y$.
  Then $\gfdiaast{R} = C \gfboxast{R} C$ and $\gfboxast{R} = C \gfdiaast{R} C$.
\end{lemma}

\begin{proof}
  $C \Lambda(\up{R}^\converse) C = \Lambda(-(C \up{R}^\converse)) = \Lambda(\up{R}^{\mathsf{d}\converse})$, using Lemmas~\ref{lemma.lambda-alpha-props}(1) and~\ref{lemma.rel-mod-props}(1) in the first step and the definition of $(-)^{\mathsf{d}}$ in the second.
  The other identity then follows using $C^2 = \Id$.
\end{proof}

The operators $\gfdiaalp{-}$, $\gfdiaast{-}$, $\gfboxalp{-}$ and $\gfboxast{-}$ specialise once again to relational $\gfdiar{-}$ and $\gfboxr{-}$, by analogy to Remark~\ref{remark.box-dia-spec}.
This can be proved not only in point-wise fashion, as the following lemma shows.
Likewise we show how (some of) the properties from Lemma~\ref{lemma.box-nu} can be derived in the language of power allegories without tabulation.

\begin{lemma}
  \label{lemma.predtrans-3}
  Let $R : X \rto \Pow Y$ and $S : X \rto Y$.
  Then
  \begin{enumerate}
  \item $\gfdiar{S} = \gfdiaalp{\Lambda(S)} = \gfdiaalp{\eta(S)}$ and $\gfboxr{S} = \gfboxalp{\Lambda(S)} = \gfboxalp{\eta(S)}$,
  \item $\gfdiaalp{R} = \gfdiaalp{\fus(R)} = \gfdiaalp{\fis(R)} $ and $\gfboxalp{R} = \gfboxalp{\fus(R)} = \gfboxalp{\fis(R)}$,
  \item $\gfdiar{S} = \gfdiaast{S{\in}} = \gfdiaast{\eta(S)}$ and $\gfboxr{S} = \gfboxast{S {\in}} = \gfboxast{\eta(S)}$,
  \item $\gfdiaalp{R} = \gfdiaast{R {\ni} {\in}} = \gfdiaast{\fis(R)} $ and $\gfboxalp{R} = \gfboxast{R {\ni} {\in}} = \gfboxast{\fis(R)}$,
  \item $\gfdiar{S} = \gfboxast{\Lambda(S)} $ and $\gfboxr{S} = \gfdiaast{\Lambda(S)}$.
  \item $\gfdiaalp{R}= \gfboxast{\fus(R)}$ and $\gfboxalp{R}= \gfdiaast{\fus(R)}$.
  \end{enumerate}
\end{lemma}

\begin{proof}
  For (1), $\gfdiaalp{\Lambda(S)} = \gfdiar{\alpha(\Lambda(S))} = \gfdiar{S}$.
  The proofs for $\eta$ and the box operators are similar.

  For (3), $\Lambda(\up{(S {\in})} {}^\converse) = \Lambda((S {\in} \Omega)^\converse) = \Lambda((S {\in})^\converse) = \Pow(S^\converse)$ using ${\in} \Omega = {\in}$.
  A proof of the second identity uses $\up{\eta(S)} = S 1 \Omega = S {\in}$ and is similar; and so are the remaining proofs.

  For (5), $\Lambda(\up{\Lambda(S)} {}^\converse) = \Lambda((\Lambda(S) \Omega)^\converse) = \Lambda({\ni} / S)$ using Lemma~\ref{lemma.rel-mod-props}(2); the second identity follows by duality.

  Items (2), (4) and (6) follow by instantiating (1), (3) and (5) with $S = \alpha(R)$.
\end{proof}

The following commutative diagram summarises the relationships between the various modal operators shown in the previous results.

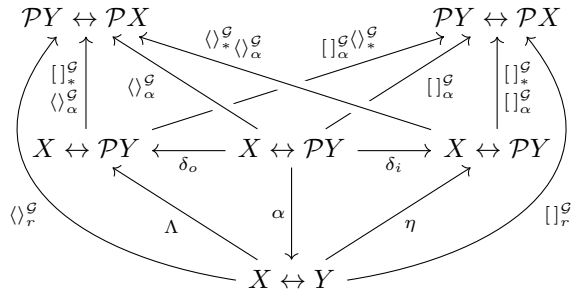
\begin{equation*}
  \begin{tikzcd}
    \Pow Y \rto \Pow X & & \Pow Y \rto \Pow X \\[4ex]
    X \rto \Pow Y \ar[u, "\gfboxast{\,}" pos=0.6, "\gfdiaalp{}" pos=0.3] \ar[urr, "\gfdiaast{}" pos=0.8, "\gfboxalp{\,}" pos=0.7] & X \rto \Pow Y \ar[l, "\fus"] \ar[ul, crossing over, "\gfdiaalp{}" pos=0.65] \ar[ur, "\gfboxalp{\,}"' pos=0.65] \ar[r, "\fis"'] \ar[d, "\alpha"'] & X \rto \Pow Y \ar[ull, crossing over, "\gfdiaalp{}"' pos=0.7, "\gfdiaast{}"' pos=0.8] \ar[u, "\gfboxalp{\,}"' pos=0.3, "\gfboxast{\,}"' pos=0.6] \\[4ex]
    & X \rto Y \ar[luu, bend left=45, out=45, in=105, end anchor=210, looseness=1.4, "\gfdiar{}"] \ar[ul, "\Lambda"] \ar[ur, "\eta"'] \ar[ruu, bend right=45, out=-45, in=-105, end anchor=330, looseness=1.4, "\gfboxr{\,}"'] &
  \end{tikzcd}
\end{equation*}


\section{Goldblatt's Axioms for Concurrent Dynamic Logic}
\label{section.goldblatt}

The box and diamond axioms of Peleg's concurrent dynamic logic have already been derived in the multirelational semantics~\cite{FurusawaStruth2015a}.
In particular, the Kleene star $R^\seq$ has been defined as the least fixpoint of $\lambda X .\ 1 \cup R \cdot X$ and studied in this setting.
It exists because this function is order-preserving on the complete lattice $(\Mult(X,X),\subseteq)$.
While the general algebra of multirelations does not satisfy the usual star axioms of Kleene algebra (owing for instance to the absence of associativity of Peleg composition), the modal star axioms of concurrent dynamic algebra can nevertheless be derived.
Relative to these results it remains to derive Goldblatt's box axioms, which consider $\fboxalp{-}$ in combination with $\fdiaast{-}$, and disregard $\fboxast{-}$.
To simplify notation, we write $\langle - \rangle$ for $\fdiaast{-}$ and $[ - ]$ for $\fboxalp{-}$ in the following.

In algebraic form, Goldblatt's axioms~\cite{Goldblatt1992} for concurrent dynamic logic are
\begin{align}
  [ R ] (P \to Q) & \subseteq ([ R ] P \to [ R ] Q), \label{eq:G1} \tag{G1} \\
  [ R ] 1 & = 1, \label{eq:G2} \tag{G2} \\
  [R \seq S] P & = [ R ] [ S ] P, \label{eq:G3} \tag{G3} \\
  [R \cup S] P & = [ R ] P \cap [ S ] P, \label{eq:G4} \tag{G4} \\
  [R \iu S] P & = (\langle R \rangle 1 \to [ S ] P) \cap (\langle S \rangle 1 \to [ R ] P), \label{eq:G5} \tag{G5} \\
  [ R^\seq ] P & \subseteq P \cap [ R ] [ R^\seq ] P, \label{eq:G6} \tag{G6} \\
  [ R^\seq ] (P \to [ R ] P) & \subseteq (P \to [ R^\seq ] P), \label{eq:G7} \tag{G7} \\
  [ P ] Q & = P \to Q, \label{eq:G8} \tag{G8} \\
  [ R ] (P \to Q) & \subseteq \langle R \rangle P \to \langle R \rangle Q, \label{eq:G9} \tag{G9} \\
  \langle R \seq S \rangle P & = \langle R \rangle \langle S \rangle P, \label{eq:G10} \tag{G10} \\
  \langle R \cup S \rangle P & = \langle R \rangle P \cup \langle S \rangle P, \label{eq:G11} \tag{G11} \\
  \langle P \iu Q \rangle P & =\langle R \rangle P \cap \langle S \rangle P, \label{eq:G12} \tag{G12} \\
  P \cup \langle R \rangle \langle R^\seq \rangle P & = \langle R^\seq \rangle P, \label{eq:G13} \tag{G13} \\
  [ R^\seq ] (\langle R \rangle P \to P) & \subseteq (\langle R^\seq \rangle P \to P), \label{eq:G14} \tag{G14} \\
  \langle P \rangle Q & = P \cap Q, \label{eq:G15} \tag{G15} \\
  [ R ] \emptyset \cup \langle R \rangle 1 & = 1, \label{eq:G16} \tag{G16}
\end{align}
where we write $P \to Q$ for $\neg P \cup Q$.

\begin{proposition}
  \label{proposition.goldblatt}
  All axioms except (\ref{eq:G3}) hold in the multirelational semantics.
\end{proposition}

\begin{proof}
  Axioms (\ref{eq:G10})-(\ref{eq:G13}) and (\ref{eq:G15}) belong to Peleg's concurrent dynamic logic.
  Algebraic variants have already been derived~\cite{FurusawaStruth2015a}.
  The remaining axioms, except (\ref{eq:G3}), have been validated with Isabelle.
  As a counterexample to (\ref{eq:G3}), consider the relation $R = \{ (a,\{a,b\}) \}$ from Example~\ref{example.alpha-counter} with $[ R \seq R ] \emptyset$.
  Then $\alpha(R \seq R) = \emptyset \subset \{ (a,a),(a,b) \} = \alpha(R) \alpha(R)$ and therefore $[ R ] [ R ] \emptyset = [ \alpha(R) \alpha(R) ] \emptyset = \{ (b,b) \} \subset \{ (a,a), (b,b) \} = [ \emptyset ] \emptyset = [ R \seq R ] \emptyset$.
\end{proof}

Goldblatt's original axioms are therefore unsound with respect to the intended multirelational semantics.
The failure of (\ref{eq:G3}) may seem surprising: after all, $(R \seq S) \seq P = R \seq (S \seq P)$ holds for all composable multirelations $R$ and $S$ and multirelational tests $P$~\cite{FurusawaStruth2015a}.
Yet the proof of Proposition~\ref{proposition.goldblatt} shows that the weak preservation of Peleg composition by $\alpha$, namely
\begin{equation*}
  \alpha(R \seq S) \subseteq \alpha(R) \alpha(S),
\end{equation*}
in Lemma~\ref{lemma.alpha-props} together with Example~\ref{example.alpha-counter} blocks any proof of (\ref{eq:G3}).
In light of Lemma~\ref{lemma.alpha-props}, at least a weak version of (\ref{eq:G3}) can be derived.

\begin{lemma}
  \label{lemma.goldblatt-repair}
  Let $R : X \rto \Pow Y$, $S : Y \rto \Pow Z$ and $P \subseteq Z$.
  Then $[ R ] [ S ] P \subseteq [R \seq S] P$.
\end{lemma}

\begin{proof}
  $\fboxalp{R} \fboxalp{S} P = \fboxr{\alpha(R)} \fboxr{\alpha(S)} P = \fboxr{\alpha(R) \alpha(S)} P \subseteq \fboxr{\alpha(R \seq S)} P = \fboxalp{R \seq S}P$, because boxes are order-reversing in their first arguments.
\end{proof}

Equality holds only in special cases such as $[ \down{R} \seq S] P = [ R ] [ S ] P$ or if $S$ is total.
The inclusion from Lemma~\ref{lemma.goldblatt-repair} should therefore replace (\ref{eq:G3}) in Goldblatt's axioms.

Goldblatt proves that his extension of propositional dynamic logic, which is a variant of Peleg's concurrent dynamic logic~\cite{Peleg1987} with the Nerode-Wijesekera box operator, is finitely axiomatisable and has the finite model property, which implies decidability.
Whether these results still hold for the sound axiomatisation using the formula in Lemma~\ref{lemma.goldblatt-repair} remains to be seen.

\begin{remark}
  The above counterexample also shows that $\fboxalp{-}$, and consequently $\fdiaalp{-}$, do not generally yield actions of multirelations on powersets with respect to $\seq$, which is atypical for modal operators.
  For relational modalities, $\fboxr{R S} = \fboxr{R} \circ \fboxr{S}$ and $\fdiar{R S} = \fdiar{R} \circ \fdiar{S}$; in fact, $\fboxr{-}$ and $\fdiar{-}$ are inclusion functors $\Rel \to \Set$.
  Similarly, for multirelational modalities, $\fdiaast{R \seq S} = \fdiaast{R} \circ \fdiaast{S}$ and $\fboxast{R \seq S} = \fboxast{R} \circ \fboxast{S}$~\cite{Peleg1987}, so that $\fboxr{-}$ and $\fdiar{-}$ are inclusion maps from the algebra of multirelations (which is not a category) into $\Set$.

  For $\fdiaalp{-}$ and $\fboxalp{-}$, by contrast, these relationships fail due to the approximative nature of $\alpha$, as shown in Proposition~\ref{proposition.goldblatt} and Lemma~\ref{lemma.goldblatt-repair}.
  In the outer and inner deterministic case, we get similar results, owing to the isomorphisms with $\Rel$.
  For instance, for $R : X \rto Y$ and $S : Y \rto Z$,
  \begin{gather*}
    \fboxalp{\Lambda(R) \seq \Lambda(S)} = \fboxalp{\Lambda(R)} \circ \fboxalp{\Lambda(S)}, \qquad
    \fdiaalp{\Lambda(R) \seq \Lambda(S)} = \fdiaalp{\Lambda(R)} \circ \fdiaalp{\Lambda(S)}, \\
    \fboxalp{\Lambda(\Id)} = \Id = \fdiaalp{\Lambda(\Id)}.
  \end{gather*}
  Further,
  \begin{gather*}
    \fboxalp{\eta(R) \seq \eta(S)} = \fboxalp{\eta(R)} \circ \fboxalp{\eta(S)}, \qquad
    \fdiaalp{\eta(R) \seq \eta(S)} = \fdiaalp{\eta(R)} \circ \fdiaalp{\eta(S)}, \\
    \fboxalp{\eta(\Id)} = \Id = \fdiaalp{\eta(\Id)}.
  \end{gather*}
\end{remark}


\section{Conclusion}

In this trilogy of articles, we have studied the inner structure of multirelations and the categories of outer and inner deterministic and univalent multirelations.
In the current article, the last in this trilogy, we have used this development to formalise Nerode and Wijesekera's alternative box operator for concurrent dynamic logic in an extension of power allegories, added a new De Morgan dual diamond and related these two operators with Peleg's modalities for multirelations and with relational modalities.
As an application, we have derived an algebraic variant of Goldblatt's axioms for concurrent dynamic logic and shown that Nerode and Wijesekera's box does not yield an action on sets, which is atypical for modal operators.

While we use a multirelational language of concrete relations and multirelations in this work, an axiomatic extension of the abstract allegorical approach, which equips boolean power allegories with multirelational operations, is the most natural continuation of this work.
The characterisation of intuitionistic modal algebras based on locally complete allegories is another interesting question.
Beyond the forward modal operators considered so far, backward modalities could be defined using our structural approach, and their application in the semantics and verification of programs or specifications with alternating nondeterminism should be explored.
The preorders $\subh$, $\subs$ and $\subem$ of multirelations seem related to the testing preorders for probabilistic processes studied by Deng et al.~\cite{DengGlabbeekHennessyMorgan2008}, so it might be interesting to see how the modalities studied in concurrent dynamic logics relate to the Hennessy-Milner modalities in their approach.

\paragraph{Acknowledgement}
Hitoshi Furusawa and Walter Guttmann thank the Japan Society for the Promotion of Science for supporting part of this research through a JSPS Invitational Fellowship for Research in Japan.


\bibliographystyle{alpha}
\bibliography{multirel}

\appendix

\section{Basis}
\label{section.basis}

Almost every operation in this paper, as well as~\cite{FurusawaGuttmannStruth2023a,FurusawaGuttmannStruth2023b} can be defined in terms of a basis of six operations that mix the relational and the multirelational language: the relational operations $-$, $\cap$, $/$ and the multirelational operations $1$, $\iu$, $\seq$.
Here we extend the list from~\cite{FurusawaGuttmannStruth2023b} with definitions of the modal operators.

\begin{multicols}{3}
\begin{itemize}
\item $R \cup S = -(-R \cap -S)$
\item $R - S = R \cap -S$
\item $\emptyset = R \cap -R$
\item $U = -\emptyset$
\item $\up{R} = R \iu U$
\item ${\in} = \up{1}$
\item $\Id = 1 / 1$
\item $R^\converse = -(-\Id / R)$
\item $S R = -(-S / R^\converse)$
\item ${\ni} = {\in}^\converse$
\item $R \backslash S = (S^\converse / R^\converse)^\converse$
\item $\syq{R}{S} = (R \backslash S) \cap (R^\converse / S^\converse)$
\item $\Lambda(R) = \syq{R^\converse}{\in}$
\item $\Pow(R) = \Lambda({\ni} R)$
\item $R_\kleisli = \Pow(R {\ni})$
\item $\mu = \Id_\kleisli$
\item $\Omega = {\in} \backslash {\in}$
\item $C = \syq{\in}{-\in}$
\item $\icpl{R} = R C$
\item $R \ii S = \icpl{(\icpl{R} \iu \icpl{S})}$
\item $\down{R} = X \ii U$
\item $\convex{R} = \up{R} \cap \down{R}$
\item $\iuone = 1 \ii \icpl{1}$
\item $\iione = \icpl{\iuone}$
\item $\dual{R} = -\icpl{R}$
\item $R \seqint S = \icpl{(R \seq \icpl{S})}$
\item $R_\seq = (\Lambda({\ni} 1) \seq 1^\converse R 1) \mu$
\item $R / S = R / S_\seq$
\item $\iuatoms = U 1$
\item $\iiatoms = \icpl{\iuatoms}$
\item $\nu(R) = R - \iuone$
\item $\tau(R) = R \cap \iuone$
\item $\alpha(R) = R {\ni}$
\item $\fis(R) = \down{R} \cap \iuatoms$
\item $\fus(R) = 1 R_\kleisli$
\item $\cofission{R} = \up{R} \cap \iiatoms$
\item $\cofusion{R} = \icpl{\fus(\icpl{R})}$
\item $\dom(R) = \Id \cap R R^\converse$
\item $\neg P = \Id - P$
\item $\fdiar{R}{P} = \dom(R P)$
\item $\fboxr{R}{P} = \neg \fdiar{R}{\neg P}$
\item $\fdiaast{R}{P} = \dom(R \seq P)$
\item $\fboxast{R}{P} = \neg \fdiaast{R}{\neg P}$
\item $\fdiaalp{R}{P} = \fdiar{\alpha(R)}{P}$
\item $\fboxalp{R}{P} = \neg \fdiaalp{R}{\neg P}$
\item $\gfdiar{R} = \Pow(R^\converse)$
\item $\gfboxr{R} = \Lambda({\ni}/R)$
\item $\gfdiaast{R} = \Lambda(\up{R} {}^\converse)$
\item $\gfboxast{R} = \Lambda(\up{R} \dual{} {}^\converse)$
\item $\gfdiaalp{R} = \gfdiar{\alpha(R)}$
\item $\gfboxalp{R} = \gfboxr{\alpha(R)}$
\item $R \subs S \Leftrightarrow S \subseteq \up{R}$
\item $R \subh S \Leftrightarrow R \subseteq \down{S}$
\item $R \subem S \Leftrightarrow R \subh S \wedge R \subs S$
\end{itemize}
\end{multicols}

If $\seq$ is extended to relations, the simpler definition $R_\seq = \Id \seq R$ may be used.
Alternatively, we could of course replace Peleg composition by Peleg lifting in the basis.
We could also replace relational intersection $\cap$ with a multirelational intersection variant $\cap$ in the basis: relational $\cap$ is obtained by $R \cap S = \alpha(R 1 \cap S 1)$ which can be defined in terms of multirelational $\cap$ and the rest of the basis.
Yet we do not know whether a multirelational $-$ could replace the relational variant.
Finally, relational $/$ is required to define some of the operations in our list as it is the only operation in the basis that can change types.
We have so far not attempted to axiomatise the basic operations in the sense of (heterogeneous) relation algebra~\cite{SchmidtStroehlein1989}, concurrent dynamic algebra~\cite{FurusawaStruth2016} or likewise.

\end{document}